\documentclass{IEEEtran}

\usepackage{graphicx}
\usepackage{amsmath}
\usepackage{amsfonts}
\usepackage{amssymb}
\usepackage{amsthm}
\usepackage{subfigure}
\usepackage{hyperref}
\usepackage{mdwlist}
\usepackage{color}
\usepackage{amsopn}
\usepackage{cite}
\usepackage{algorithm}
\usepackage{algorithmic}

\newtheorem{theorem}{Theorem}[section]
\newtheorem{proposition}[theorem]{Proposition}

\theoremstyle{definition}

\newtheorem{definition}{Definition}

\theoremstyle{remark}
\newtheorem{remark}{Remark}

\newcommand{\chg}[1]{\textcolor{black}{#1}}
\newcommand{\ecp}[1]{\textcolor{black}{#1}}
\newcommand{\argmin}{\arg\!\min}
\newcommand{\argmax}{\arg\!\max}

\newcommand{\TheTitle}{Minimum Probabilistic Finite State Learning Problem on Finite Data Sets: Complexity, Solution and Approximations}

\floatstyle{ruled}
\newfloat{Program}{htb}{lop}[section]
\floatname{Program}{Algorithm Description}

\title{{\TheTitle}}

\author{
Elisabeth Paulson\thanks{E. Paulson is with the Operations Research Center, Massachusetts Institute of Technology, Cambridge, MA, 02142, \href{mailto:elisabethpaulson63@gmail.com}{elisabethpaulson63@gmail.com}}~and~Christopher Griffin\thanks{C. Griffin is with the Mathematics Department, United States Naval Academy, Annapolis, MD, 21402, \href{mailto:griffinch@ieee.org}{griffinch@ieee.org}}
}

\begin{document}

\maketitle
\begin{abstract} In this paper, we study the problem of determining a minimum state probabilistic finite state machine capable of generating statistically identical symbol sequences to samples provided. This problem is qualitatively similar to the classical Hidden Markov Model problem and has been studied from a practical point of view in several works beginning with the work presented in: Shalizi, C.R., Shalizi, K.L., Crutchfield, J.P. (2002) \textit{An algorithm for pattern discovery in time series.} Technical Report 02-10-060, Santa Fe Institute. arxiv.org/abs/cs.LG/0210025. We show that the underlying problem is $\mathrm{NP}$-hard and thus all existing polynomial time algorithms must be approximations on finite data sets. Using our $\mathrm{NP}$-hardness proof, we show how to construct a provably correct algorithm for constructing a minimum state probabilistic finite state machine given data and empirically study its running time.
\end{abstract}


\section{Introduction}

The goal of the field of computational mechanics is to use statistical methods to produce models that explain an observed, and possibly probabilistic, behavior. In \cite{CS02}, Shalizi and Crutchfield give an excellent overview of computational mechanics. These models can be used to simulate the behavior of, predict future performance for, and gain insights into the underlying processes that govern the observed behavior \cite{CS02}. Crutchfield broadly discusses the idea of pattern recognition and the importance of discerning signal and noise in \cite{Crutch12}. This is the underlying motivation behind computational mechanics-- attempting to produce the most meaningful models to capture the important features of a behavior.

Perhaps the most common and well-known models for pattern recognition are Hidden Markov models. Hidden Markov models (HMMs) are a pattern recognition tool used to construct a Markov model when the state space is unknown. HMMs are used for a wide variety of purposes, such as speech recognition \cite{Rab89}, handwriting recognition \cite{Xue06,Hu96}, and tracking \cite{Lef03}. Given a process which produces some string of training data, there are many algorithms that are widely used to infer a Hidden Markov model for the process. In this paper, we focus on the approach developed by Shalizi \cite{Shal02} for producing $\epsilon$-machines from a string of input data, which can be thought of as a kind of HMM. This work is extended in \cite{BSG09b,BSG10} and \cite{LSCBG13a}. A model like the one discussed in this paper is also assumed in \cite{BSG09,BSC+10,LSCBG13}

Shalizi's approach to constructing a HMM is to find statistically significant groupings of the training data which then correspond to causal states in the HMM. In this formulation, each state is really an equivalence class of unique strings. This algorithm groups unique strings based on the conditional probabilities of the next symbol as a window slides over the training data. The window gradually increases in length up to a maximum length $L$, which is the maximum history length that contains predictive power for the process. This approach is called the Causal State Splitting and Reconstruction (CSSR) algorithm. The result of the CSSR algorithm is a Markov model where each state consists of a set of histories of up to length $L$ that all share the same conditional probability distributions on the next symbol.

In this paper we merge states based on a multiple comparison test of the conditional probability distributions of each state; however, it should be noted that there are many accepted state merging techniques. Stolcke and Omohundro propose a state merging technique for Hidden Markov Models that uses Bayesian posterior probabilities to merge similar states \cite{Stolcke93}. Another known merging technique is the subtree merging algorithm \cite{Crutch89}. This is a very similar technique to that proposed in this paper— states are merged together based on the conditional probabilities of the next symbol. However, due to the iterative nature of this procedure, it has the same flaw as the CSSR algorithm in that it is not guaranteed to produce a minimal state machine. A new approach to estimating epsilon-machines uses a finite sample to generate a set of candidate unifilar hidden Markov Models, and uses Bayesian Structural Inference to determine the posterior probability of each model topology \cite{Strelioff14}. \cite{Schm06} eschews unifilar models to obtain an algorithm with linear running time using a Euclidean metric on conditional symbol distributions.

Assuming an input string of infinite length where the conditional probability distributions converge to their true values, the CSSR algorithm produces a minimal-state $\epsilon$-machine for the given process. That is, any time two states could be combined while still maintaining the desired properties of the $\epsilon$-machine, they are. \chg{As the length of the input string approaches infinity, the CSSR algorithm is correct in the limit as a result of the Glivenko-Cantelli Theorem \cite{tucker1959}.  However, for input sequences of finite length, the CSSR algorithm does not guarantee a minimal-state machine due to the technique by which the strings are grouped.} Consequently, state explosion can occur. This paper seeks to develop a new approach, based on the CSSR algorithm, which always guarantees a minimal-state HMM. We use a very similar idea to the CSSR algorithm, but formulate it as an integer programming problem whose goal is to minimize the number of states. We then show that solving this integer program is NP-hard.

\chg{It should be noted that there are many hardness results for problems that are similar to ours. In \cite{Abe}, Abe et al. show that polynomial time learnability of stochastic rules, in particular probabilistic concepts, with respect to the Kullback-Leibler divergence is equivalent to the same notion using quadratic distance. Abe et al. give a polynomial time learning algorithm for a class of convex linear combinations of stochastic rules. Pitt and Warmuth also consider a slightly different problem than this paper by considering DFA’s whose goal is to learn a set of positive and negative words and then classify incoming words into these categories \cite{Pitt}. It was shown by Gold that determining the smallest possible consistent DFA for this problem is $\mathrm{NP}$-hard \cite{Gold78}, and Pitt and Warmuth show that determining an approximately small DFA is also hard \cite{Pitt}. In \cite{KV94} the authors study the problem of distribution free learning of Boolean functions. In particular, they cover the case of learning unifilar acyclic finite automata from string samples. In particular, the results presented in this paper are analogous and consistent with our results. Where we differ is in considering probabilistic generating systems leading to probabilistic finite state machines. Additionally, the inclusion of the state minimization problem means we must consider cyclic automata. }

\chg{The constraints of this integer programming problem ensure that we obtain a model that is statistically consistent with, and able to generate, the input string. Thus, it accomplishes the same generalizability and consistency as the CSSR algorithm. As both Shalizi and Crutchfield explain, the best model is the one with minimum size that also minimizes the amount of error and maximizes predictive power \cite{CS02, Crutch12}. In this paper, we make strides at accomplishing that goal by producing a consistent, generalizable model, that also guarantees state minimization.} We then provide a reformulation of the integer programming problem as a minimal clique covering problem, which yields a faster algorithm (in practice) for finding the minimal-state HMM.

\section{Notation and Preliminaries}
Let $\mathcal{A}$ be a finite set of symbols or alphabet representing the observable events of a given system. Let $\mathbf{y}$ be the symbolic output of our system, so $\mathbf{y}$ is a sequence composed of the elements in $\mathcal{A}$. We will write $\mathbf{y}=y_1y_2...y_n$ to denote the individual elements in $\mathbf{y}$, so $y_i \in \mathcal{A}$.

Let $\mathcal{A}^*$ be the set of all sequences composed of symbols from $\mathcal{A}$, and let $\lambda$ be the empty sequence. If $\mathbf{y}$ is a sequence, then $\mathbf{x}$ is a \emph{subsequence}
if
\begin{enumerate*}
\item $\mathbf{x}$ is a sequence with symbols in $\mathcal{A}$ and
\item There are integers $i$ and $k$ such that
    $\mathbf{x}=y_{i}y_{i+1} \cdots y_{i+k}$.
\end{enumerate*}

The subsequences of $\mathbf{y}$ of length $k$ constitute the sliding data windows of length $k$ over $\mathbf{y}$.

For us, a hidden Markov model is a tuple $G=\langle{Q,\mathcal{A},\delta,
p}\rangle$, where $Q$ is a finite set of states, $\mathcal{A}$ is a
finite alphabet, $\delta \subseteq Q \times \mathcal{A} \times Q$ is
a transition relation, and $p:\delta
\rightarrow [0,1]$ is a probability function such that
\begin{equation}
\sum_{a \in \mathcal{A}, q' \in Q} p(q,a,q') = 1 \;\;\;\forall q \in Q
\end{equation}

This is slightly different than the standard definition of an HMM as observed in \cite{Rab89} because we are particularly interested in the state transition relation and associated probabilities, rather than the symbol probability distribution for a given state.

The Baum-Welch Algorithm is the standard expectation maximization
algorithm used to determine the transition relation and probability function of a
hidden Markov model. However, this algorithm requires an
initial estimate of the transition structure, so some initial knowledge of the structure of the Markov process
governing the dynamics of the system must be known.

Given a sequence $\mathbf{y}$ produced by a stationary process, the Causal
State Splitting and Reconstruction (CSSR) Algorithm
infers a set of causal states and a transition structure for
a hidden Markov model that provides a maximum likelihood estimate of the
true underlying process dynamics. In this case, a causal state is a function mapping histories to their equivalence classes, as in \cite{Shal04}.

The states are defined as equivalence classes of conditional
probability distributions over the next symbol that can be generated by
the process. The set of states found in this manner accounts for the
unifilar behavior of the process while tolerating random noise that
may be caused by either measurement error or play in the system under
observation. The CSSR Algorithm has useful information-theoretic
properties in that it attempts to maximize the mutual information among states
and minimize the remaining uncertainty (entropy) \cite{Shal04}.

The CSSR Algorithm is straightforward. We are given a sequence
$\mathbf{y} \in \mathcal{A}^*$ and know a priori the value $L
\in \mathbb{Z}_+$.\footnote{There are some heuristics for choosing $L$ found in \cite{Shal04}} For values of $i$ increasing from $0$ to $L$, we
identify the set of sequences $W$ that are subsequences of
$\mathbf{y}$ and have length $i$. (When $i=0$, the empty string is
considered to be a subsequence of $\mathbf{y}$.) We compute the
conditional distribution on each subsequence $\mathbf{x} \in W$ and
partition the subsequences according to these distributions.  These
partitions become states in the inferred HMM. If states already exist,
we compare the conditional distribution of subsequence $\mathbf{x}$ to
the conditional distribution of an existing state and add $\mathbf{x}$
to this state, if the conditional distributions are ruled identical.
Distribution comparison can be carried out using either a
Kolmogorov-Smirnov test or a $\chi^2$ test with a specified level of
confidence. The level of confidence chosen affects the Type I error
rate.  Once state generation is complete, the states are further split
to ensure that the inferred model has a unifilar transition
relation $\delta$.  Algorithm \ref{alg:CSSR} provide pseudocode for the CSSR Algorithm.

\begin{algorithm}[htbp]
\scriptsize
\caption{-- CSSR Algorithm}
\textbf{Input:} Observed sequence $\mathbf{y}$; Alphabet $\mathcal{A}$, Integer $L$\\
\underline{Initialization:}
\begin{algorithmic}[1]
\STATE Define state $q_0$ and add $\lambda$ (the empty string) to state
    $q_0$. Set $Q=\{q_0\}$.
\STATE Set $N:=1$.
\end{algorithmic}

\underline{Splitting} (Repeat for each $i \leq L$)
\begin{algorithmic}[1]
\STATE Set $W=\{\mathbf{x} | \exists q \in Q (\mathbf{x} \in q \wedge
    |\mathbf{x}| = i - 1)\}$ \COMMENT{The set of strings in states of the current
    model with length equal to $i-1$}
\STATE Let $N$ be the number of states.
\FOR{each $\mathbf{x} \in W$}
	\FOR{each $a \in \mathcal{A}$}
		\IF {$a\mathbf{x}$ is a subsequence of $\mathbf{y}$}
			\STATE Determine $f_{a\mathbf{x}|\mathbf{y}}: \mathcal{A} \rightarrow [0,1]$, the probability distribution over the next input symbol.
    			\STATE Let $f_{q_j|\mathbf{y}}: \mathcal{A} \rightarrow [0,1]$ be the unified state conditional probability distributions; that is, the probability given the system is in state $q_i$, that the next symbol observed will be $a$. For each $j$, compare $f_{q_j|\mathbf{y}}$ with $f_{a\mathbf{x}|\mathbf{y}}$ using an appropriate statistical test with confidence level $\alpha$. Add $a\mathbf{x}$ to the state that has the most similar probability distribution as measured by the $p$-value of the test. If all tests reject the null hypothesis that $f_{q_j|\mathbf{y}}$ and $f_{a\mathbf{x}|\mathbf{y}}$ are the same, then create a new state $q_{N+1}$ and add $a\mathbf{x}$ to it. Set $N:=N+1$.
		\ENDIF
	\ENDFOR
\ENDFOR
\end{algorithmic}

\underline{Reconstruction}
\begin{algorithmic}[1]
\STATE Let $N_{0} = 0$.
\STATE Let $N$ be the number of states.
\WHILE{$N_{0} \neq N$}
	\FOR{each $i \in 1,\dots,N$}
		\STATE Set $k:=0$
   		\STATE Let $M$ be the number of sequences in state $q_i$. Choose a sequence $\mathbf{x}_0$ from state $q_i$. 			\STATE Create state $p_{ik}$ and add $\mathbf{x}_0$ to it
		\FOR{all sequences $\mathbf{x}_{j}$ $(j>0)$ in state $q_i$}
			\STATE For each $a \in \mathcal{A}$ $\mathbf{x}_{j}a$ produces a sequence that is resident in another state $q_{k}$. Let $\delta(\mathbf{x}_j,a)=q_k$.
			\STATE For $l=0,\dots,k$, choose $\mathbf{x}$ from $p_{ik}$. If $\delta(\mathbf{x}_j,a)=\delta(\mathbf{x},a)$ for all $a \in \mathcal{A}$, then add $\mathbf{x}_j$ to $p_{ik}$. Otherwise, create a new state $p_{ik+1}$ and add $\mathbf{x}_j$ to it. Set $k:=k+1$.
		\ENDFOR
	\ENDFOR
	\STATE Reset $Q = \{p_{ik}\}$; recompute the state conditional probabilities $f_{q|\mathbf{y}}$ for $q \in Q$ and assign transitions using the $\delta$ functions defined above.
	\STATE Set $N_{0} = N$.
	\STATE Set $N$ to be the number of states in our current model.
\ENDWHILE
\STATE The model of the system has state set $Q$ and transition probability function computed from the $\delta$ functions and state conditional probabilities.
\end{algorithmic}
\label{alg:CSSR}
\end{algorithm}

The complexity of CSSR is $O(k^{2L+1})+ O(N)$, where $k$ is the size
of the alphabet, $L$ is the maximum subsequence length considered, and
$N$ is the size of the input symbol sequence\cite{Shal04}. Given a stream of
symbols $\mathbf{y}$, of fixed length $N$, from alphabet
$\mathcal{A}$, the algorithm is linear in the length of the input data
set, but exponential in the size of the alphabet.


\subsection{A Practical Problem with CSSR}
As observed in Shalizi and Shalizi (2004), the CSSR converges to the minimum state estimator asymptotically \cite{Shal04}, however it does not always result in a correct minimum state estimator for a finite sample.  With an infinitely long string, all of the maximum likelihood estimates converge to their true value and the CSSR algorithm works correctly. That is, as $N\rightarrow \infty,$ the probability of a history being assigned to the wrong equivalence class approaches zero. A formal proof of this is given in \cite{Shal02} and relies on large deviation theory for Markov chains.

The error in estimating the conditional probabilities with strings of finite length can cause the CSSR algorithm to produce a set of causal states that is not minimal.

The following example will clarify the problem. Consider $\mathcal{A}=\{0,1\}$ and $\mathbf{y}$ defined as follows:
\begin{equation}
\begin{aligned}
y_i=0 \;\; &1 \leq i \leq 518 \\
[y_{i}\; y_{i+1}\;y_{i+2}\; y_{i+3}] =1100 \;\; &518<i\leq 582\\
[y_{i}\;y_{i+1}\; y_{i+2}] =100 \;\; &582<i \leq 645\\
[y_{i}\;y_{i+1}\; y_{i+2}] =101 \;\; & 645<i \leq 648
\end{aligned}
\label{eqn:sequence}
\end{equation}

Without loss of generality, we consider exclusively strings of length two. The strings, in order of appearance, are $\{00,01,11,10\}$. The conditional probability distribution on the next element for each string is shown in Table \ref{table:probabilities}:
\begin{table}[ht]
\centering
\caption{Conditional Probabilities}
\begin{tabular}{c c c}
\hline
string & $Pr(0 | \text{string})$ & $Pr(1 | \text{string})$ \\
\hline \hline
00 & 0.9314 & 0.0686 \\
01 & 0.5789 & 0.4211 \\
11 & 1.0 & 0 \\
10 & 0.9737 & 0.0263 \\
\hline
\end{tabular}
\label{table:probabilities}
\end{table}

The $p$-values associated with comparing each string's conditional probability distribution to that of every other string are shown in Table \ref{table:pvalues}:
\begin{table}[ht]
\centering
\caption{p-values}
\begin{tabular}{| c || c | c | c | c|}
\hline
string & 00 & 01 & 11 & 10 \\ \hline \hline
00 & 1.0000 & -- & -- & -- \\ \hline
01 & 0 & 1.0000 & -- & -- \\ \hline
11 & 0.0067 & 0 & 1 & -- \\ \hline
10 & 0.0944 & 0 & 0.7924 & 1 \\
\hline
\end{tabular}
\label{table:pvalues}
\end{table}

We now step through the Splitting phase of the CSSR algorithm:
\begin{enumerate}
\item String $00$ is put into its own state, $q_1$
\item We use Pearson's Chi-Squared test to test if the $01$ and $00$ have the same conditional distribution. This results in a $p$-value of approximately 0, as seen Table \ref{table:pvalues} so string $01$ is put into a new state, $q_2$
\item We compare the conditional distribution of $11$ to string $00$. This results in a $p$-value of 0.0067, and comparing it to $01$ results in a p-value of 0, so string $11$ is put into a new state, $q_3$
\item We compare the conditional probability distribution of $10$ to that of $00$. This results in a $p$-value of 0.0944. Comparing it to $00$ results in a $p$-value of 0, and comparing it to $11$ results in a p-value of 0.7924. Since 0.7924 is greater than 0.0944 and greater than the chosen significance level of 0.05, the string $10$ is clustered with string $11$, so $q_3=\{11,10\}$.
\end{enumerate}
Thus, at the end of the Splitting step of CSSR, we are left we three different states:
\begin{align*}
q_1=\{00\}\\
q_2=\{01\}\\
q_3=\{11,10\}
\end{align*}

We now go on to the Reconstruction step. Let $\mathbf{x}_{ij}$ be a state in $q_i$ where $j \in 1,2,...,|q_i|$. The Reconstruction step checks whether, for each $a \in \mathcal{A}$, $\delta(\mathbf{x}_{i1},a)=\delta(\mathbf{x}_{i2},a)=...=q_k$. Recall that $\delta(\mathbf{x}_{ij},a)=q_k$ means that $\mathbf{x}_{ij}a$ produces a sequence that resides in state $q_k$.  If this is not satisfied, then state $q_i$ must be broken up into two or more states until we have a unifilar transition relation function.

In our example, the first two states do not need to be checked since they each only consist of one string. We check the third state:
$\delta(11,0)=q_3$ since string $10 \in q_3$, and $\delta(10,0)=q_1$ since string $0 \in q_1$. Thus, determinism is not satisfied and state $q_3$ must be split into two states. The result of the CSSR algorithm is a four-state Markov chain where each string resides in its own state.

We now show why this result is not the minimal state Markov generator for this input sequence. Suppose that during the Splitting step, string four was put into state $q_1$ instead of $q_3$. This could have been done soundly from a statistical point of view, since 0.0944 is also greater than our alpha-level of 0.05. However, by the CSSR heuristic, the fourth string was grouped according to the \textit{highest} p-value. If the fourth string were put in $q_1$, then after the Splitting step we would have obtained the following states:
\begin{align*}
q_1=\{00,10\}\\
q_2=\{01\}\\
q_3=\{11\}
\end{align*}
Now we move on to the Reconstruction step. This time, we only need to consider state $q_1$. Notice that $\delta(10,0)=\delta(00,0)=q_1$ and $\delta(00,1)=\delta(10,1)=q_2$. We see that state $q_1$ does satisfy determinism, and does not need to be split. Thus, our final number of states is only three, which is minimal.

This provides us with motivation to reformulate the CSSR algorithm so that it always produces a the minimal number of states, even when given a small sample size. \chg{The remainder of this paper seeks to address this issue, proposes a reformulation of the CSSR algorithm which always results in minimal state models even with a finite sample, and discusses the computational complexity of the minimum state unifilar probabilistic finite state automaton modeling problem.}

\subsection{State Explosion with CSSR}
Although in the preceding example, the number of states that the CSSR algorithm returned was only one larger than the integer program, the differences become larger as we increase the size of the alphabet, $k$. In this section we generalize the preceding example to show that state explosion can occur quadratically as a function of $k$, while keeping $L=2$.

\chg{In the sequel, we treat individual $L$-length strings as vertices in a graph whose edge set consists of pair of vertices with statistically equivalent next-symbol distributions.} In the example, the problem with CSSR arose from a connected component of this graph that contained three vertices of the form $xy$ $zy$ and $uv$ ($00$, $10$ and $11$ in the example). The CSSR algorithm clustered $xy$ and $uv$ together, resulting in the state being split, whereas the integer program clustered $xy$ and $zy$ together, satisfying the unifilar property. The ``worst case'' string for the CSSR algorithm is one that maximizes the number of these special three-vertex components, and then chooses the wrong clustering for each of them.

This worst case scenario results in CSSR producing a machine with $k^2$ states. The integer program we discuss in Section \ref{sec:IP}, however, would choose correctly for each component so that no states are split during the splitting phase. With some simple calculation it can be seen that, for an even $k$, the integer program will result in a machine with $\frac{k^2}{2}+\lceil \frac{k^2}{6} \rceil $ states, and an odd $k$ will result in a machine with $\frac{(k-1)k}{2}+\lceil \frac{(k^2-3k)}{6} \rceil +k$ states. A figure showing a comparison between CSSR and the integer program for varying values of $k$ can be seen in Figure \ref{fig:state_comparison}.

\begin{figure}[htbp]
	\centering
	\includegraphics[scale=0.5]{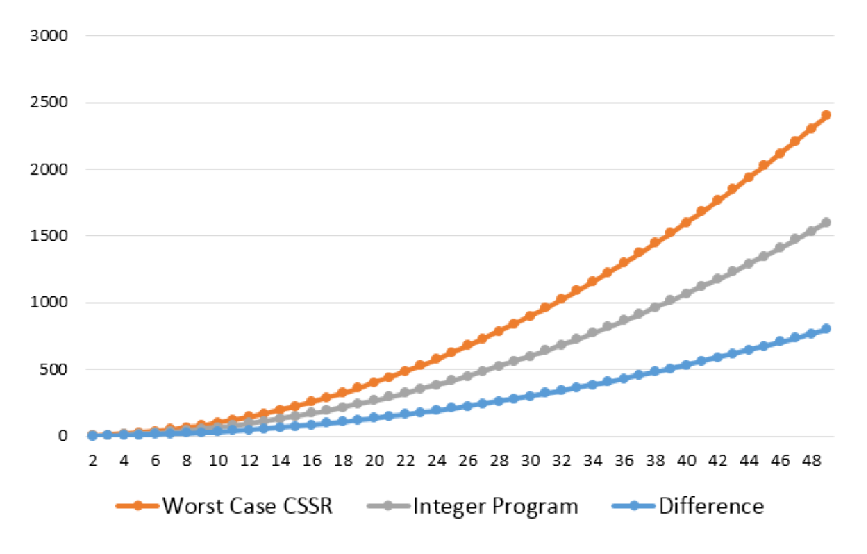}
	\caption{The number of states produced by CSSR and the integer program for the "worst case" strings}
	\label{fig:state_comparison}
\end{figure}

It remains to see that there are indeed strings which have this ``worst case'' property of their graph counterpart having the maximum possible number of special three-vertex components. In order to show this, we argue that given any conditional probability distribution on all strings of length $L$, we can build a sequence whose strings of length $L$ have statistically equivalent distributions to the specified distribution.

A formal proof of this is a part of the proof of our general NP-hardness result. Intuitively, this can be accomplished by \chg{constructing} a Markov chain where each string of length $L$ corresponds to a state and creating probabilistic transitions as needed. We then use this chain to produce a sequence. A long enough sequence will have statistically equivalent properties to a desired distribution, and can thus be constructed as a ``worst case'' string.

\section{Integer Programming Formulation}\label{sec:IP}

Much of this section has been adapted from \cite{Clus13}. Suppose we have a string $\mathbf{y} \in \mathcal{A}^*$ and we assume that $L$ is known. Let $W$ be the set of distinct substrings of length $L$ in $\mathbf{y}$, and let $|W|=n$. When we cluster the elements of $W$ into states we must be sure of two things: 1) That all strings in the same state have the same conditional distributions, and 2) that the resulting transition function is unifilar. Our goal is to formulate an integer programming problem that minimizes the number of states and preserves the two conditions above. Define the following binary variables:
\begin{equation}
x_{ij}= \begin{cases}
1 &\text{string $i$ maps to state $j$} \\
0 & \text{else} \end{cases}
\label{eq:x}
\end{equation}

\begin{remark}We assume implicitly that there are $n$ states, since $n$ is clearly the maximum number of states that could be needed. If, for some $j$, $x_{ij}=0\;\forall \; i$, then that state is unused, and our true number of states is less than $n$.
\label{rem:nstates}
\end{remark}

Let:

\begin{equation}
z^\sigma_{il}= \begin{cases}
1 &\text{string $i$ maps to string $l$ upon symbol $\sigma$} \\
0 & \text{else} \end{cases}
\label{eq:z}
\end{equation}

For example, $z_{il}^\sigma=1$ if $i=\langle x_1, ..., x_L \rangle$ and $l=\langle x_2,...,x_L,\sigma \rangle$. Assuming $j$ and $k$ are state indices while $i$ and $l$ are string indices, we also define the variable $y_{jk}^\sigma$ \chg{by the following logical rule:}
$$(x_{ij}=1) \wedge (z_{il}^\sigma=1) \wedge (x_{lk}=1) \implies (y_{jk}^\sigma=1)$$

This variable ensures the encoded transition relation maps each state to the correct next state given a certain symbol. Specifically, if string $i$ in clustered in state $j$ and string $i$ transforms to string $l$ given symbol $\sigma$, and string $l$ is clustered in state $k$, then $(j,\sigma,k)$ must be the transition relation. This can be written as the constraint
$$(1-x_{ij})+(1-z_{il}^\sigma)+(1-x_{lk})+y_{jk}^\sigma \geq 1 \;\;\; \forall i,j,k,l,\sigma$$

In order for the automata to be unifilar, we must have the condition that
$$\sum_k y_{jk}^\sigma \leq 1 \;\;\; \forall j,\sigma $$
This condition ensures that for a given symbol, each state can transition to only one other state.

We ensure that the strings in each state have the same conditional probability distributions using a parameter $\mu$, where
\begin{equation}
\mu_{il}= \begin{cases}
1 &\text{\small strings $i$ and $l$ have identical conditional distributions} \\
0 & \text{else} \end{cases}
\label{eq:u}
\end{equation}
In order to determine if two strings have identical distributions an appropriate statistical test (like Pearson's $\chi^2$ or Kolmogorov-Smirnov) must be used. The variables must satisfy
$$(x_{ij}=1) \wedge (x_{lj}=1) \implies (\mu_{il}=1)$$
This states that strings $i$ and $l$ can only belong to the same state if they have statistically indistinguishable distributions. This can be written as the following constraint
$$(1-x_{ij})+(1-x_{lj})+\mu_{il} \geq 1 \;\;\; \forall i,l,j $$

Finally, we define the variable $p$ to be used in our objective function. \chg{From Remark \ref{rem:nstates}}, this variable simply enumerates up the number of ``used" states out of $n$. That is:
\begin{equation}
p_{j}= \begin{cases}
0 &\sum_{i} x_{ij}=0\\
1 & \text{else} \end{cases}
\label{eq:p}
\end{equation}
This can be enforced as the following constraint
$$p_j \geq \frac{\sum_i x_{ij}}{n}, \; p_j \in \{0,1\}$$

Note that $\frac{\sum_i x_{ij}}{n} \leq 1$. These constraints are equivalent to Equation \ref{eq:p} because  when $\sum_i x_{ij}=0$,  $ p_j$ will be 0, and when $\sum_i x_{ij}>0$, $p_j$ will be 1, when we minimize $\sum_j p_j$ in our optimization problem. This addition extends the work done in \cite{Clus13}, in which a series of optimization problems had to be solved.

\begin{definition}[Minimum State Deterministic pFSA Problem]
The following binary integer programming problem, which, if feasible, defines a mapping from strings to states and a probability transition function between states which satisfies our two requirements and is called the Minimum State Deterministic pFSA (MSDpFSA) Problem.
\begin{equation}
 P(N) = \left\{
\begin{aligned}
\min\;\; &\sum_{j\in \{1,...,n\}} p_j \\
s.t.\;\; &(1-x_{ij})+(1-z_{il}^\sigma)+(1-x_{lk})+y_{jk}^\sigma \geq 1, \\
&\hspace*{3em}\forall i,l \in W,j,k\in  \{1,...,n\},\sigma \in \mathcal{A} \\
&\sum_k y_{jk}^\sigma \leq 1 \;\;\; \forall j,k\in  \{1,...,n\},\sigma \in \mathcal{A} \\
&(1-x_{ij})+(1-x_{lj})+\mu_{il} \geq 1,\\
&\hspace*{3em} \forall i,l\in W,j\in  \{1,...,n\} \\
&p_j \geq \frac{\sum_i x_{ij}}{n} \;\; \forall i\in W,j\in  \{1,...,n\} \\
&\sum_j x_{ij}=1 \;\; \forall i\in W \\
&p_j \in \{0,1\} \;\; \forall j \in  \{1,...,n\}\\
& x_{ij} \in \{0,1\} \;\; \forall i\in W, j\in  \{1,...,n\}\\
&y_{jk}^\sigma \in \{0,1\} \;\; \forall j,k\in  \{1,...,n\},\sigma \in \mathcal{A} 
\end{aligned}\right.
\label{eqn:Problem}
\end{equation}
\end{definition}
We note that in Problem \ref{eqn:Problem}, the constraint $(1-x_{ij})+(1-x_{lj})+\mu_{il} \geq 1$ becomes redundant $\mu_{il}=1$, and similarly the constraint $(1-x_{ij})+(1-z_{il}^\sigma)+(1-x_{lk})+y_{jk}^\sigma \geq 1$ is redunant when $z_{il}=1$.  

The following proposition is clear from the construction of Problem \ref{eqn:Problem}:
\begin{proposition} Any optimal solution to MSDpFSA yields an encoding of a minimum state probabilistic finite machine capable of generating the input sequence(s) with statistically equivalent probabilities. \chg{Further,
\begin{enumerate}
\item There is a starting state and set of transitions that produces the input sequence exactly
\item Additional strings produced by a solution to MSDpFSA are statistically equivalent to the input string
\item The number of states in any solution to MSDpFSA is bounded above by the number of unique strings of length $L$ in the input sequence
\end{enumerate}}
\end{proposition}

\begin{remark}
Because this formulation does not rely on the assumption of infinite string length to correctly minimize the number of states, it succeeds where the CSSR algorithm fails. Instead of clustering strings into the state with the \textit{most} identical conditional probability distributions, this optimization problem uses the identical distributions as a condition with the \textit{goal} of state minimization. Thus, in the example, even though the fourth string had a higher p-value when compared to the third state than the first state, since both of the p-values are greater than 0.05 this algorithm allows the possibility of clustering the fourth string into either state, and chooses the state that results in the smallest number of final states after reconstruction. Unlike CSSR, this algorithm does not rely on asymptotic convergence of the conditional probabilities to their true value. However, it is clear that as the sample size increases the probability distributions will become more exact, leading to better values of $\mu$.
\end{remark}

\subsection{Resolution to the Example Problem}
Using the integer program formulation to solve the example previously presented does result in the minimum state estimator of the process represented by the given input string. The integer program finds the following solution for the variables and parameters $z,x,y, \mu$, and $p$ seen in Table \ref{table:variables}. Our objective function is to minimize the sum of the variables $p_j$. As can be seen in Table \ref{table:variables}, there is only one nonzero column in the $x$ matrix (the last column). Thus, only $p_4=0$, and the the value of our objective function is 3.

\begin{table}
\centering
\caption{Integer Program Variables and Parameters}
\subtable[$x$ values]
{\begin{tabular}{|c| c c c c| }
\hline
-- & $q_1$ & $q_2$ & $q_3$ & $q_4$ \\ \hline
00 & 1 & 0 & 0 & 0\\
01 & 0 & 1 & 0 & 0\\
11 & 0 & 0 & 1 & 0\\
10 & 1 & 0 & 0 & 0 \\ \hline
\end{tabular}}
\subtable[$\mu$ values]
{\begin{tabular} {|c| c c c c|}
\hline
-- & 00 & 01 & 11 & 10 \\ \hline
00 & 1 & 0 & 0 & 1 \\
01 & 0 & 1 & 0 & 0 \\
11 & 0 & 0 & 1 & 1 \\
10 & 1 & 0 & 1 & 1 \\ \hline
\end{tabular}}
\subtable[$z^0$ values]
{\begin{tabular} {|c| c c c c|}
\hline
-- & 00 & 01 & 11 & 10 \\ \hline
00 & 1 & 0 & 0 & 0 \\
01 & 0 & 0 & 0 & 1 \\
11 & 0 & 0 & 0 & 1 \\
10 & 1 & 0 & 0 & 0 \\ \hline
\end{tabular}}
\subtable[$z^1$ values]
{\begin{tabular} {|c| c c c c|}
\hline
-- & 00 & 01 & 11 & 10 \\ \hline
00 & 0 & 1 & 0 & 0 \\
01 & 0 & 0 & 1 & 0 \\
11 & 0 & 0 & 1 & 0 \\
10 & 0 & 1 & 0 & 0 \\ \hline
\end{tabular}}
\subtable[$y^0$ values]
{\begin{tabular} {|c| c c c|}
\hline
-- & $q_1$ & $q_2$ & $q_3$  \\ \hline
$q_1$ & 1 & 0 & 0 \\
$q_2$ & 1 & 0 & 0 \\
$q_3$ & 1 & 0 & 0 \\ \hline
\end{tabular}}
\subtable[$y^1$ values]
{\begin{tabular} {|c| c c c|}
\hline
-- & $q_1$ & $q_2$ & $q_3$  \\ \hline
$q_1$ & 0 & 1 & 0 \\
$q_2$ & 0 & 0 & 1 \\
$q_3$ & 0 & 0 & 1 \\ \hline
\end{tabular}}
\subtable[$p$ values]
{\begin{tabular} {|c c c c|}
\hline
1 & 1 & 1 & 0 \\ \hline
\end{tabular}}
\label{table:variables}
\end{table}

We can also recover a matrix of transition probabilities shown in Table \ref{table:T}. Our final reduced state space machine is shown in Figure \ref{fig:machine}.

\begin{table}
\centering
\caption{State Transition Probabilities}
\begin{tabular}{|c| c c c|}
\hline
-- & $q_1$ & $q_2$ & $q_3$  \\ \hline
$q_1$ & 0.9341 & 0.0659 & 0 \\
$q_2$ & 0.5789 & 0 & 0.4211 \\
$q_3$ & 1 & 0 & 0 \\ \hline
\end{tabular}
\label{table:T}
\end{table}

\begin{figure}[htbp]
\centering
\includegraphics[scale=1]{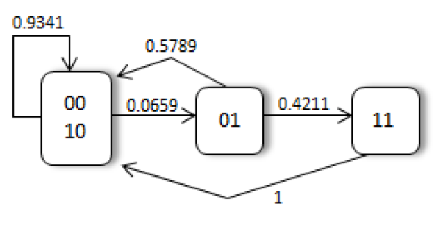}
\caption{The reduced state space machine derived from the input given by Equation \ref{eqn:sequence}
}
\label{fig:machine}
\end{figure}

\section{Computational Complexity of MSDpFSA}
In this section we assert that the integer programming problem given by Problem \ref{eqn:Problem} is $\mathrm{NP}$-hard.

We show that \chg{MSDpFSA} is $\mathrm{NP}$-hard by reduction to the minimum clique covering problem. Recall given a graph, a minimal clique covering problem is to identify a set of cliques in the graph so that each vertex belongs to at least one clique and so that this set of cliques is minimal in cardinality. \chg{We construct this proof in stages: (i) To illustrate the idea behind the reduction, we show that the Minimum State Non-Deterministic pFSA Problem, in which we drop the determinism constraint, is $\mathrm{NP}$-hard by reduction to the minimum clique covering problem. (ii) We show for an aribtrary graph structure, we can create a data \textit{gadget} that transforms the minimum clique covering problem into MSDpFSA in polynomial time.}

\subsection{Minimum State Non-Deterministic pFSA Problem}
While the CSSR algorithm attempts to find a minimal state \textit{unifilar} finite-state automata, being unifilar is not a necessary property for accurate reconstruction. For a given input sequence, the construction of a finite-state automata which is \textit{not} necessarily unifilar is less computationally intensive than solving the MSDpFSA problem, and is discussed by Schmiedekamp et al. (2006) \cite{Schm06}. Like the CSSR algorithm, the algorithm presented in \cite{Schm06} is a heuristic which does not always results in a minimal state probabilistic FSA as we show in the sequel. The following definition provides an integer programming problem for optimally solving a minimal state pFSA which is not necessarily unifilar. This is identical to Problem \ref{eqn:Problem} except that we discard the constraints that ensure the automata is unifilar.

\begin{definition}[Minimum State Non-Deterministic pFSA Problem]
The following binary integer programming problem, which, if feasible, defines a mapping from strings to states and a probability transition function between states which is not necessarily unifilar and is called the Minimum State Non-Deterministic pFSA (MSNDpFSA) Problem.
\begin{equation}
 P'(N) = \left\{
\begin{aligned}
\min\;\; &\sum_{j\in  \{1,...,n\}} p_j \\
s.t.\;\; &(1-x_{ij})+(1-x_{lj})+\mu_{il} \geq 1,\\
&\hspace*{4em}\forall i,l\in W,j\in  \{1,...,n\} \\
&p_j \geq \frac{\sum_i x_{ij}}{n} \;\; \forall i\in W,j\in  \{1,...,n\} \\
&\sum_j x_{ij}=1 \;\; \forall i\in W\\
& p_j \in \{0,1\} \;\; \forall j \in  \{1,...,n\} \\
& x_{ij} \in \{0,1\} \;\; \forall i\in W, j\in \{1,...,n\}
\end{aligned}\right.
\label{eqn:MSNDpFSA}
\end{equation}
\end{definition}

\begin{theorem} MSNDpFSA is $\mathrm{NP}$-Hard.
\label{prop:NPhardness}
\end{theorem}
\begin{IEEEproof}
Let $G=(V,E)$ be the graph on which we want to find the minimal clique covering. Assume that $V=\{1,2,...,n\}$ and let $I$ be a $n\times n$ matrix such that $I_{ij}=1 \iff$ there is an edge connecting vertices $i$ and $j$. We reduce Problem \ref{eqn:MSNDpFSA} by letting $n$ be the number of unique strings of length $L$ in $\mathbf{y}$, so we can let each string correspond to a vertex of $G$. Let $I_{ij}=1 \iff \mu_{ij}=1$. This means that two strings are connected if and if only if they have identical conditional probability distributions. We show that Problem \ref{eqn:MSNDpFSA} is equivalent to finding a minimal clique cover of $G$ where $\sum_j p_j$ is the number of cliques. Let the set of cliques corresponding to the minimal clique cover of $G$ be $C=\{C_1,...,C_m\}$, where $C_j$ is a specific clique, and $C_j=V_j \subset V$.

We can define the variables $x$ by $x_{ij}=1 \iff V_i \in C_j$. Thus, the constraint that $(1-x_{ij})+(1-x_{lj})+\mu_{il} \geq 1 \;\;\; \forall i,l,j $ simply means that if two vertices are in the same clique then there must be an edge between them. We also have that $p_j=0$ iff there is at least one vertex in clique $j$, so the set of $j$ such that $p_j=1$ corresponds to the non-empty cliques, i.e., is identical to $C$ (since $C$ only consists of non-empty cliques). Thus, minimizing $\sum_j p_j$ is equivalent to minimizing the number of cliques needed to cover $G$. The constraint that $\sum_j x_{ij}=1 \;\; \forall i$ simply means that each vertex belongs to exactly one clique. Thus, it is clear that the integer programming problem given in Problem \ref{eqn:MSNDpFSA} is equivalent to a minimal clique covering and is $\mathrm{NP}$-hard.
\end{IEEEproof}


We illustrate the equivalence of Problem \ref{eqn:MSNDpFSA} and the minimal clique covering through an example. Using the same example as before, the resulting graph $G$ is shown in Figure \ref{fig:graph}. The string $00$ and $10$ have an edge in common because they have identical conditional distributions as noted by Table \ref{table:pvalues}. By the same reasoning, $10$ and $11$ also have an edge in common. However, $00$ and $11$ do not share an edge, and $01$ has no incident edges.

The integer programming problem given by Equation \ref{eqn:MSNDpFSA} produces either of the following two results:
\begin{gather*}
Q=\{q_1=\{00,10\},  q_2=\{11\}, q_3=\{01\}\} \;\; \text{or} \\
Q'=\{q_1=\{11,10\}, q_2=\{00\},q_3=\{01\}\}
\end{gather*}
Both of these two groupings results in one of two minimal clique coverings. Let $V_1=00, \; V_2=10,\; V_3=11,\; V_4=01$.
\begin{gather*}
C=\{C_1=\{1,2\}, \; C_2=\{3\},\; C_3=\{4\}\} \;\; \text{or} \\
C'=\{C_1=\{2,3\}, \; C_2=\{1\},\; C_3=\{4\}\}
\end{gather*}

\begin{figure}[htbp]
\centering
\subfigure[]{\label{fig:graph}
  \includegraphics[scale=1]{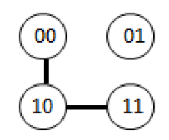}}
\subfigure[]{\label{fig:graph2}
  \includegraphics[scale=1]{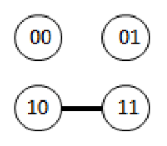}}
\caption{(a) The graph $G$ corresponding to the string $y$ given in Equation \ref{eqn:sequence}. (b) The graph $G$ considered by the CSSR algorithm.}
\end{figure}

When the unifilar constraint is added to the integer programming formulation, the result is $Q'$ (or equivalently $C'$) instead of $Q$ (or $C$). The solver recognizes that $Q \;(or \; C)$ would have to be split in order to satisfy the unifilar property, so $Q' \; (or \; C')$ is chosen instead. If we think of the CSSR algorithm in terms of our graph equivalency, the CSSR algorithm does not consider the existence of an edge between $00$ and $10$. Let $T$ be Table \ref{table:pvalues}. For a vertex $i$,  CSSR only places an edge between $i$ and $\argmax_{j\neq i} \{T(i,j): T(i,j)>0.05\}$.Thus, by the CSSR algorithm, the graph is really the image shown in Figure \ref{fig:graph2}. In this formulation, there is only one minimal state clustering (minimal clique covering), $G \;(or\;C)$ which does not satisfy the unifilar property, so $G$ is split into four states.

We now state and prove the main result of this paper.
Like the proof of Theorem \ref{prop:NPhardness}, the proof of the following theorem relies on reduction to the Minimum Clique Cover problem, but the reduction gadget is 
more complex to construct.
\begin{theorem} MSDpFSA is $\mathrm{NP}$-hard.
\end{theorem}
\begin{IEEEproof}
%
%
%
	We show that for any arbitrary graph $G$ there is an input sequence $\mathbf{y}$ so that the solution to the Minimum State Deterministic pFSA (MSDpFSA) problem generated from this string solves the minimum clique covering problem for $G$. It follows at once that MSDpFSA is $\mathrm{NP}$-hard.


	Let $G=(V,E)$ be an arbitrary graph with vertices $V$ and edges $E$. Since $G$ is arbitrary, it may be composed of several components. Two vertices $v_1,v_2 \in V$ are in the same \textit{component} of $G$ if there is a walk from $v_2$ to $v_1$.



	As before, we will map the vertices of $G$ to unique sub-strings of length $L$ taken from the input sequence $\mathbf{y}$.
	%
	%
	Let $n_0$ be the cardinality of the largest component in $G$, let $c(G)$ be the number of components and let $N = \max\{n_0,c(G)\}$. Let $L = 2$ and assume we have alphabet $\mathcal{A} = \{0,\dots,N\}$. There are $n_0 + 1$ strings of length $2$ with form $\{s0 : s \in \{0,\dots,n_0 - 1\}\}$. For example, the strings $00$,$10$, $20$ etc. all satisfy this pattern. To each vertex in the largest component of $G$ (with size $n_0$) assign one of these strings.

	Let $n_1$ be the size of the second largest component and repeat this process of string assignment but with the set of length $2$ strings $\{s1 : s \in \{0,\dots,n_1 - 1\}$. It is clear we can repeat this process until each vertex in $V$ has been assigned a length $2$ string.

	Since $L=2$ and our alphabet has $N+1$ symbols, there are $(N+1)^2$ unique strings. For each possible string that was \textit{not} already assigned to a vertex, add a vertex to $G$ corresponding to this string to create a new graph $G'$. Call these added vertices $V'$. Clearly $G$ is a subgraph of $G'$ and any minimum clique covering of $G$ produces a minimum clique covering of $G'$ since the additional isolated vertices are necessarily in their own clique.

	Notice any clique in $G'$ not made up of vertices in $V'$ must be a subset of a component in $G$. Let $\delta:V \cup V' \times \mathcal{A} \rightarrow V \cup V'$ be the induced transition function on the graph $G'$ when we associate each vertex with its corresponding string and let $c(v)$ denote the component in which $v$ resides. By construction, $\delta$ has the property that if $\sigma \in \mathcal{A}$ and $v_1,v_2 \in V$ and $c(v_1) = c(v_2)$, then $\delta(v_1,\sigma) = \delta(v_2,\sigma)$. For example every string (vertex) in the set $\{00,10,\cdots(n_0 - 1)0\}$ transitions to the string $0\sigma$ on symbol $\sigma$. Therefore by construction, every possible clique on $G'$ must transition to one and only one other clique. Thus, any minimal clique covering is guaranteed to generate a unifilar transition function. In particular, on $\sigma$, this component will map either to an extra string or to the next $\sigma^\text{th}$ component.


	It remains to show that we can construct a sequence polynomial time using the alphabet $\mathcal{A}$ so that the conditional probabilities of each string of length two ensure that we obtain the correct edges to produce the graph $G'$ and thus $G$. If so, then MSDpFSA reduces tot the minimum clique cover problem.

	Because the equivalence of two probability distributions in the CSSR algorithm is defined through a statistical test, the distributions only need to match within a predefined value $\epsilon$ in order to be deemed statistically equivalent. This value of $\epsilon$ will vary by statistical test, sample size, etc. but it exists.

	To construct the string, define a Markov chain so that each state corresponds to the vertex set of $V \cup V'$, where vertices are treated as their length 2-string assignments. Suppose $v$ corresponds to string $\alpha_1\alpha_2$, we enforce a transition structure that ensures only transitions of the form $\alpha_1\alpha_2 \mapsto \alpha_2\alpha_3$ may have non-zero probability $p_{3|12}$. The resulting substring produced by traversing these states is $\alpha_1\alpha_2\alpha_3$ and asymptotically will approach $p_{3|12}$ as the Markov chain is used to generate the string. By judiciously building the transition probabilities, we may ensure that in an appropriate statistical test, the string $\mathbf{y}$ that is generated by the Markov chain will produce conditional probabilities satisfying the following constraints:
	\begin{align*}
	&|p_{\sigma | v_1} - p_{\sigma | v_2}| \leq \epsilon \quad
	\forall v_1,v_2 \in E, \forall \sigma \in \mathcal{A}\\
	&\sum_{\sigma \in \mathcal{A}}|p_{\sigma | v_1} - p_{\sigma | v_2}| \geq |\mathcal{A}|\epsilon \quad
	\forall v_1,v_2 \not\in E\\
	&\sum_{\sigma \in \mathcal{A}} p_{\sigma | v} = 1 \quad
	\forall v \in V\\
	&p_{\sigma | v} \geq 0 \quad \forall v \in V, \sigma \in \mathcal{A}
	\end{align*}
	The second constraint ensures there is at least one symbol $\sigma \in \mathcal{A}$ so that $|p_{\sigma|v_1} - p_{\sigma|v_2}| \geq \epsilon$.

	For small enough $\epsilon$ (implying a potentially large sample length), the feasible region is non-empty and thus there is an assignment of probabilities satisfying these constraints. We now show that this assignment can be accomplished in polynomial time. Note first, we can re-write these constraints as:
	\begin{align*}
	&p_{\sigma | v_1} - p_{\sigma | v_2} \leq \epsilon \quad
	\forall v_1,v_2 \in E, \forall \sigma \in \mathcal{A}\\
	&p_{\sigma | v_2} - p_{\sigma | v_1} \leq \epsilon \quad
	\forall v_1,v_2 \in E, \forall \sigma \in \mathcal{A}\\
	&\sum_{\sigma \in \mathcal{A}}|p_{\sigma | v_1} - p_{\sigma | v_2}| \geq |\mathcal{A}|\epsilon \quad
	\forall v_1,v_2 \not\in E\\
	&\sum_{\sigma \in \mathcal{A}} p_{\sigma | v} = 1 \quad
	\forall v \in V\\
	&p_{\sigma | v} \geq 0 \quad \forall v \in V, \sigma \in \mathcal{A}
	\end{align*}
	eliminating one non-linear term and leaving only the (now) third (non-linear) constraint. For all $v_1,v_2 \not\in E$ and $\sigma \in \mathcal{A}$, define:
	\begin{equation}
	p_{\sigma | v_1} - p_{\sigma | v_2} = z_{\sigma|v_1v_2}^+ - z_{\sigma|v_1v_2}^-
	\label{eqn:trick}
	\end{equation}
	where $z_{\sigma|v_1v_2}^+, z_{\sigma|v_1v_2}^- \geq 0$. Consider the linear programming problem:
	\begin{equation}
	\left\{
	\begin{aligned}
	\min \;\; & \sum_{v_1,v_2 \not\in E, \sigma \in \mathcal{A}}\left( z_{\sigma|v_1v_2}^+ + z_{\sigma|v_1v_2}^-\right)\\
	s.t.\;\;& p_{\sigma | v_1} - p_{\sigma | v_2} \leq \epsilon \quad
	\forall v_1,v_2 \in E, \forall \sigma \in \mathcal{A}\\
	& p_{\sigma | v_2} - p_{\sigma | v_1} \leq \epsilon \quad
	\forall v_1,v_2 \in E, \forall \sigma \in \mathcal{A}\\
	&\sum_{\sigma \in \mathcal{A}}\left(z_{\sigma|v_1v_2}^+ + z_{\sigma|v_1v_2}^-\right) \geq |\mathcal{A}|\epsilon \quad
	\forall v_1,v_2 \not\in E\\
	&p_{\sigma | v_1} - p_{\sigma | v_2} = z_{\sigma|v_1v_2}^+ - z_{\sigma|v_1v_2}^-,\\
	&\hspace*{6em} \forall v_1,v_2 \not\in E, \forall \sigma \in \mathcal{A}\\
	&\sum_{\sigma \in \mathcal{A}} p_{\sigma | v} = 1 \quad
	\forall v \in V\\
	&p_{\sigma | v}, z_{\sigma|v_1v_2}^+, z_{\sigma|v_1v_2}^- \geq 0 \quad \forall v \in V, \sigma \in \mathcal{A}
	\end{aligned}
	\right.
	\label{eqn:LP}
	\end{equation}
	We know that Equation \ref{eqn:trick} must hold. Further, the objective function is minimized if and only if either $z_{\sigma|v_1v_2}^+ > 0$ or $ z_{\sigma|v_1v_2}^- > 0$, but not both. Therefore, $(z_{\sigma|v_1v_2}^+ + z_{\sigma|v_1v_2}^-) = |p_{\sigma | v_1} - p_{\sigma | v_2}|$.  Thus, the constraints of the \textit{linear programming problem} (Expression \ref{eqn:LP}) are equivalent to the constraints required by our probability assignment problem. Linear programming problems are solvable in polynomial time by \cite{Karmarker84}. Furthermore, the interior point methods that are known to be polynomial will return feasible points very close to extreme points, implying that we can find a $p_{\sigma|v} > 0$. (It is actually sufficient to execute the optimization only until the original constraints are satisfied.) After transition probability assignment, remove any states that do not have in-transitions, these are unreachable. The fact that $p_{\sigma|v} > 0$  ensures the irreducibility of the Markov chain. For any given state $xy$, there is non-zero probability of transitioning to $yz$, where $z$ is an arbitrary element of $\mathcal{A}$. Furthermore, $yz$ has a non-zero probability of transitioning to $zw$, where $w$ is also an arbitrary element of $\mathcal{A}$. Thus, the chain is irreducible because every state can be reaching by every other state. Strings produced by the resulting Markov chain will have statistical properties that yield a sub-string statistical equivalence graph like $G'$.
	Thus we have found a polynomial time reduction of MSDpFSA to the Minimum Clique Covering Problem. Therefore, the MSDpFSA is $\mathrm{NP}$-hard.
\end{IEEEproof}

\section{Minimal Clique Covering Reformulation of CSSR Algorithm}

In this section we present an algorithm for determining the minimal state hidden Markov model using a minimal clique vertex covering reformulation. As shown in the previous section, the two problem formulations are equivalent if we exclude the unifilar constraint. Let $\mathbf{W}=\{S_1,...,S_n\}$ be the set of unique strings of length $L$ in $\mathcal{A}$, so $W$ is the set of vertices in our graph formulation. In the revised algorithm, we first use the CSSR Algorithm, Algorithm \ref{alg:CSSR}, to find an upper bound on the number of cliques needed. We then use Bron-Kerbosch algorithm to enumerate all maximal cliques, $C=\{C_1,C_2,...,C_m\}$, given as Algorithm \ref{alg:Bron} \cite{Bron73}.

Define $H:=\argmin_R \{|R| :  R \subset \mathbf{C} \; \land \;  \cup_i R_i=W\}$. Thus $H$ is the set of maximal cliques of minimal cardinality such that every vertex belongs to at least one clique. This can be found using a simple binary integer programming problem given in Algorithm \ref{alg:clique}. We can think of Algorithm \ref{alg:clique} as defining a mapping from every string to the set of all maximal cliques. A clique is ``activated" if at least one string is mapped to it. The goal, then, is to activate as few cliques as possible. Of course, each string can only be mapped to a clique in which it appears. We define a few variables, which are present in our integer programming formulation: \\ \\
\begin{equation}
y_{i}= \begin{cases}
1 &C_i \in H\\
0 & \text{else} \end{cases}
\label{eq:y}
\end{equation}
The variable $y$ can be thought of as whether or not clique $C_i$ is activated.\\
\begin{equation}
I_{ij}= \begin{cases}
1 &S_i \in C_j \in H\\
0 & \text{else} \end{cases}
\label{eq:I}
\end{equation}
\begin{equation}
s_{ij}= \begin{cases}
1 &S_i \text{ is mapped to } C_j \in H\\
0 & \text{else} \end{cases}
\label{eq:s}
\end{equation}
To distinguish between $I$ and $s$, notice that each string is only  mapped to one clique, but each string could appear in more than one clique.

\begin{proposition}
The set $H$ is a minimal clique vertex covering.
\label{prop:covering}
\end{proposition}
\begin{IEEEproof}
Suppose there is a minimal clique vertex covering of a graph $G$ with $k$ cliques such that every clique is not maximal. Choose a non-maximal clique and expand the clique until it is maximal. Continue this procedure until every clique is maximal. Let our set of cliques be called $R$. Clearly $|R|=k$ since no new cliques were created. Also note that $R \subset C$ since $R$ consists of maximal cliques. Further, $\cup_i R_i=W$ since we started with a clique vertex covering. Finally, notice that $R=\argmin_H \{|H| :  H \subset \mathbf{C} \; \land \;  \cup_i H_i=W\}$ because $|R|=k$ and $k$ is the clique covering number of $G$, so it is impossible to find an $H$ with cardinality less than $k$.
\end{IEEEproof}

\begin{algorithm}[htbp]
\footnotesize
\caption{--Finding all maximal cliques}
\textbf{Input:} Observed sequence $\mathbf{y}$; Alphabet $\mathcal{A}$, Integer $L$\\
\textbf{Setup:}
\begin{algorithmic}[1]
\STATE Call the set of strings that appear in $\mathbf{y}$ $\mathbf{W}=\{S_1,...,S_n\}$, so $\mathbf{W}$ is our set of vertices
\STATE Determine $f_{\mathbf{x}|\mathbf{y}}(a)$ for each unique string
\STATE Generate an incident matrix, $\mathbf{U} \in \mathbb{R}^{nxn}$, $\mathbf{U}_{ij}=u_{ij}$, where $u_{ij}$ is defined in \ref{eq:u}
\STATE $P \leftarrow [1,2,...,n]$ \COMMENT{$P$ holds the prospective vertices connected to all vertices in $R$}
\STATE $R \leftarrow \emptyset$ \COMMENT{Holds the currently growing clique}
\STATE $X \leftarrow \emptyset$ \COMMENT{Holds the vertices that have already been processed}
\end{algorithmic}
\textbf{function} Bron-Kerbosch(R,P,X)
\begin{algorithmic}[1]
	\STATE $k \leftarrow 0$
	\IF{$P=\emptyset$  and $X=\emptyset$}
		\STATE report $R$ as a maximal clique, $R=C_k$
		\STATE $k \leftarrow k+1$
	\ENDIF
	\FOR{each vertex $v$ in $P$}
		\STATE Bron-Kerbosch($R \cup v, P \cap N(v), X \cap N(v)$) \COMMENT{where $N(v)$ are the neighbors of $v$}
		\STATE $P \leftarrow P \setminus v$
		\STATE $X \leftarrow X \cup v$
	\ENDFOR
	\STATE \textbf{return} $\{C_1,C_2,...,C_m\}$
\end{algorithmic}
\label{alg:Bron}
\end{algorithm}

\begin{algorithm}[htbp]
\footnotesize
\caption{--Minimal Clique Vertex Covering}
\textbf{Input:} Observed sequence $\mathbf{y}$; Alphabet $\mathcal{A}$, Integer $L$\\
\textbf{Set-up:}\\
\underline{1. CSSR}
\begin{algorithmic}
\STATE run the CSSR Algorithm (Algorithm \ref{alg:CSSR})
\STATE \textbf{return} $k$, the number of cliques found by CSSR
\end{algorithmic}
\underline{2. Find all maximal cliques}
\begin{algorithmic}
\STATE run Algorithm \ref{alg:Bron}
\STATE \textbf{return} $\{C_1,C_2,...,C_m\}$
\end{algorithmic}
\textbf{Finding all minimal clique coverings:}
\begin{algorithmic}
\STATE Let $I_{ij}=1$ denote string $S_i$ belonging to clique $C_j$
\STATE Let $s_{ij}=1$ denote string $S_i$ being mapped to clique $C_j$
\end{algorithmic}
$$Q(C)=\left\{
\begin{aligned}
\min\;\; &\sum_{j\in \{1,...,m\}} y_j \\
&\hspace*{4em}\{\text{$y_j$ indicates whether clique $C_j$ is being used or not}\}\\
s.t.\;\; &y_j \geq \frac{\sum_i s_{ij}}{|C_j|} \;\; \forall i\in \{1,...,n\}, j \in \{1,...,m\} \\
&\sum_{j} y_j \leq k \;\; \forall j \in \{i,..,m\} \\
& s_{ij} \leq I_{ij} \;\; \forall i \in \{1,...,n\}, j \in \{1,...,m\}\\
& \sum_j s_{ij}=1 \;\; \forall i \in \{1,...,n\}, j \in \{1,...,m\}
\end{aligned}\right.$$
\label{alg:clique}
\end{algorithm}

\begin{proposition}
Any solution to Q(C) results in a minimal clique vertex covering.
\end{proposition}
\begin{IEEEproof}
This is clear by noting that $\argmin Q(C)$ results in the subset, $H$, of $C=\{C_1,C_2,...,C_m\}$ where $H$ is defined as $\argmin_R \{|R| :  R \subset \mathbf{C} \; and \;  \cup_i R_i=W\}$ in conjunction with Proposition \ref{prop:covering}.
\end{IEEEproof}

Before discussing the rest of the algorithm, which concerns the unifilar property, we first state an important remark about Algorithm \ref{alg:clique}. It is possible that the set $H$ is such that some vertices are in more than one maximal clique. However, we are actually interested in the set of minimal clique vertex coverings for which each vertex only belongs to \textit{one} clique, which can be extracted from $H$. See Figures \ref{fig:minclique} and \ref{fig:cover} for an example. The set $H$ is shown in Figure \ref{fig:minclique}. From this set $H$, we can deduce the two minimal clique vertex coverings shown in \ref{fig:cover}. In initializing Algorithm \ref{alg:clique reconstruction}, which imposes the unifilar property on each covering, we find the set of all minimal coverings for which each vertex belongs to one clique.

Once we have determined the set of minimal clique covers of our original graph where each vertex only belongs to one clique, we select a final clique covering which is minimal and unifilar. This is done by considering each minimal clique covering individually and then restructuring it when necessary to be unifilar. This corresponds to the reconstruction part of Algorithm \ref{alg:CSSR}. Note in Algorithm \ref{alg:clique reconstruction} each vertex corresponds to a string of length $L$, which came from an initial string $\mathbf{y}$. Also recall that we have previously defined $\mathbf{z}$ as

\[ z^\sigma_{il}= \begin{cases}
1 &\text{string $i$ maps to string $j$ upon symbol $\sigma$} \\
0 & \text{else} \end{cases}
\]

where $\sigma$ is any element of our alphabet $\mathcal{A}$. The following proposition is clear from the construction of Algorithms \ref{alg:clique} and \ref{alg:clique reconstruction}.

\begin{proposition}
Algorithm \ref{alg:clique} along with Algorithm \ref{alg:clique reconstruction} produces an encoding of a minimum state probabilistic finite machine capable of generating the input sequence.
\end{proposition}

\begin{figure}[htbp]
\centering
\includegraphics[scale=0.75]{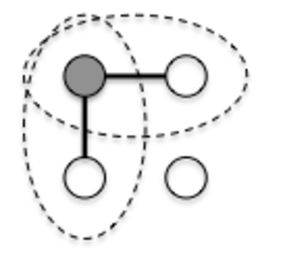}
\caption{An example output from Algorithm \ref{alg:clique}
}
\label{fig:minclique}
\end{figure}

\begin{figure}[htbp]
\centering
\begin{subfigure}
	\centering
	\includegraphics[scale=.75]{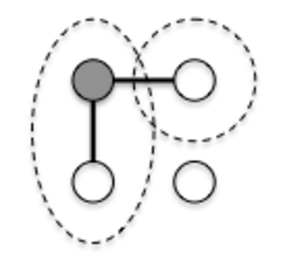}
	\label{fig:cover1}
\end{subfigure}
\begin{subfigure}
	\centering
	\includegraphics[scale=.75]{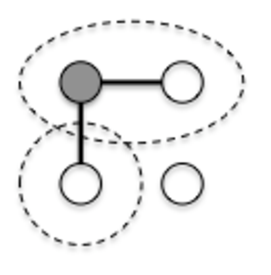}
	\label{fig:cover2}
\end{subfigure}
\caption{Minimal clique coverings where each vertex is only in one clique}
\label{fig:cover}
\end{figure}

\begin{algorithm}[htbp]
\footnotesize
\caption{-- Minimal Clique Covering Reconstruction}
\textbf{Input:} Observed sequence $\mathbf{y}$; Alphabet $\mathcal{A}$, Integer $L$; \\
\textbf{Set-up:}
\begin{algorithmic}
\STATE Perform Algorithm \ref{alg:Bron} on $\mathbf{y}$ and obtain a minimal clique covering
\STATE From this initial clique covering, find the set of all minimal coverings such that each vertex belongs to exactly one clique, an example of which is seen in Figure \ref{fig:cover}. Let this set of minimal coverings be $\textbf{T}=\{T_1,...,T_l\}$
\STATE From $\mathbf{y}$, $\mathcal{A}$, and the set $\{S_1,...,S_n\}$ as found in Algorithm \ref{alg:Bron}, determine the matrix $z$.
\end{algorithmic}
\textbf{Minimal Clique Covering Reconstruction:}
\begin{algorithmic}[1]
\FOR{$h=1$ \TO $l$} 
	\STATE $i=1$
	\STATE $N=$ the number of cliques in each $T_h$ \COMMENT note that all $T_h$ have the same number of cliques because they are all minimal
	\STATE $M=N$
	 \WHILE{$i \leq M$}
		\STATE $N=M$
		\STATE $F=$ the set of all vertices (strings) which are in clique $i$
		\IF{length$(F)>1$}
			\FOR{$j=2,3,...,\#\{F\}$}
			\STATE Use $\textbf{z}$ to find what clique $F(j)$ maps to for each $\sigma \in \mathcal{A}$
			\IF{$F(j)$ does not map to the same clique as $F(1)$ upon each $\sigma$}
				\IF{$M>N$ and $F(j)$ maps to the same clique as some $F(k)$ upon each $\sigma$, $k \in [N+1,M]$}
					\STATE Add $F(i)$ to the clique containing $F(k)$
				\ELSE
					\STATE Create a new clique containing $F(i)$
					\STATE $M=M+1$
				\ENDIF
			\ENDIF
			\ENDFOR
		\ENDIF
	\STATE $i=i+1$
	\ENDWHILE
\ENDFOR
\STATE $\text{FinalCovering}=\min\{C_h | C_h\text{ has the minimum number of columns } \forall h\}$
\end{algorithmic}
\textbf{return} FinalCovering
\label{alg:clique reconstruction}
\end{algorithm}

\section{Comparing run times of modified CSSR algorithms}

This paper has discussed three different approaches for determining minimal state hidden Markov models from a given input sequence: the CSSR algorithm, CSSR as an integer programming problem, and the minimal clique covering reformulation. We now compare the run times of all three algorithms. We find that the CSSR algorithm has the fastest run time, however it does not always produce a minimal state model as we have seen in a prior example. The minimal clique covering reformulation is relatively fast and always results in a minimal state model. The integer programming formulation is extremely slow, but also always results in a minimal state model. This makes CSSR a useful heuristic for solving the $\mathrm{NP}$-hard MSDpFSA problem.

All computational experiments in this section were performed using MATLAB, and the binary integer programming problems were solved using MATLAB's \textit{bintprog} function. The formulations were each tested using randomly generated strings in order to observe their behavior on a diverse set of inputs. The CSSR integer programming formulation, using only a two-character alphabet and \chg{a randomly generated} sequence of length ten as input, takes about 100 seconds to run. We can see in Figure \ref{fig:cliquetimes} that the minimal clique covering reformulations takes less than 0.2 seconds for \chg{randomly generated two-character} sequences up to length 10,000. For only two character alphabets, the run time appears to be a nearly linear function of the sequence length, with spikes occurring presumably in cases where more reconstruction was needed. In Figure \ref{fig:cliquetimes} we see the run time of the minimal clique covering reformulation for a \chg{randomly generated string with a three-character alphabet}. For sequences up to length 100, the algorithm took no more than 120 seconds, which is remarkably better than the integer programming problem.

\begin{figure}[htbp]
\centering
\begin{subfigure}
	\centering
	\includegraphics[scale=.4]{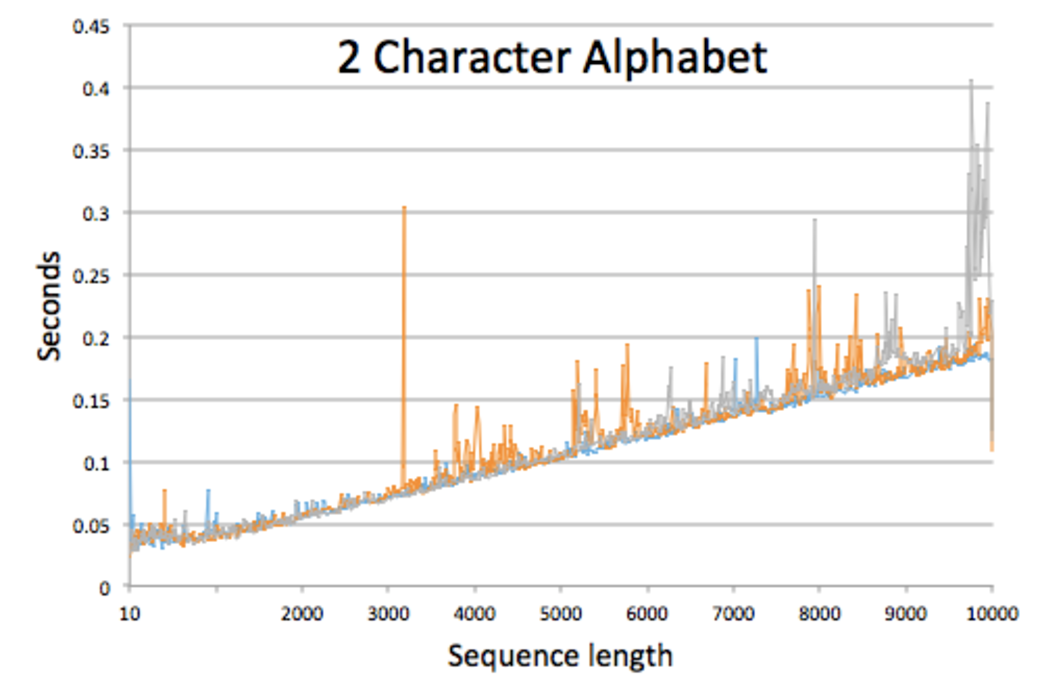}
\end{subfigure}
\begin{subfigure}
	\centering
	\includegraphics[scale=.36]{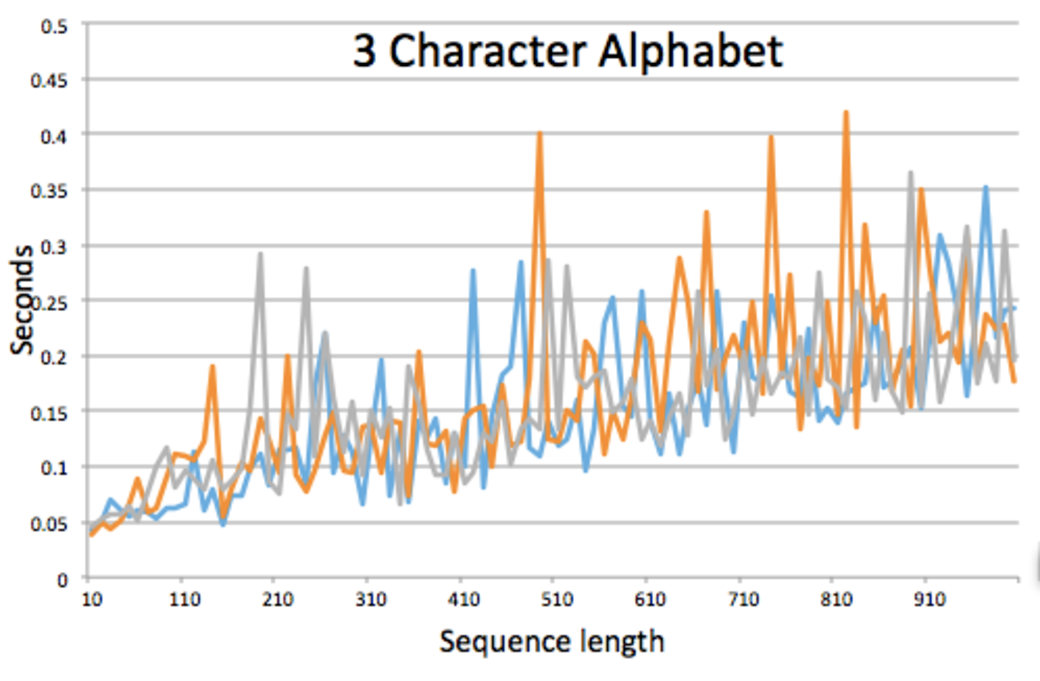}
\end{subfigure}
\caption{Clique covering reformulation run times. \chg{Each graph shows three different trial runs of randomly generated strings.}}
\label{fig:cliquetimes}
\end{figure}

For two-, three-, and 4-character alphabets, we also compare the minimal clique covering formulation time to that of the original CSSR algorithm. This can be seen in Figure \ref{fig:charactercompare} where we take the average of the trial run times for the CSSR and clique covering formulation to compare the two. Note that with a two-character alphabet, the clique covering formulation is slightly faster, but for the three- and four-character alphabets the CSSR algorithm is significantly faster. This is due to the fact that the CSSR algorithm does not actually guarantee a minimal state Markov model.
\begin{figure}[h]
\centering
\begin{subfigure}
	\centering
	\includegraphics[scale=.38]{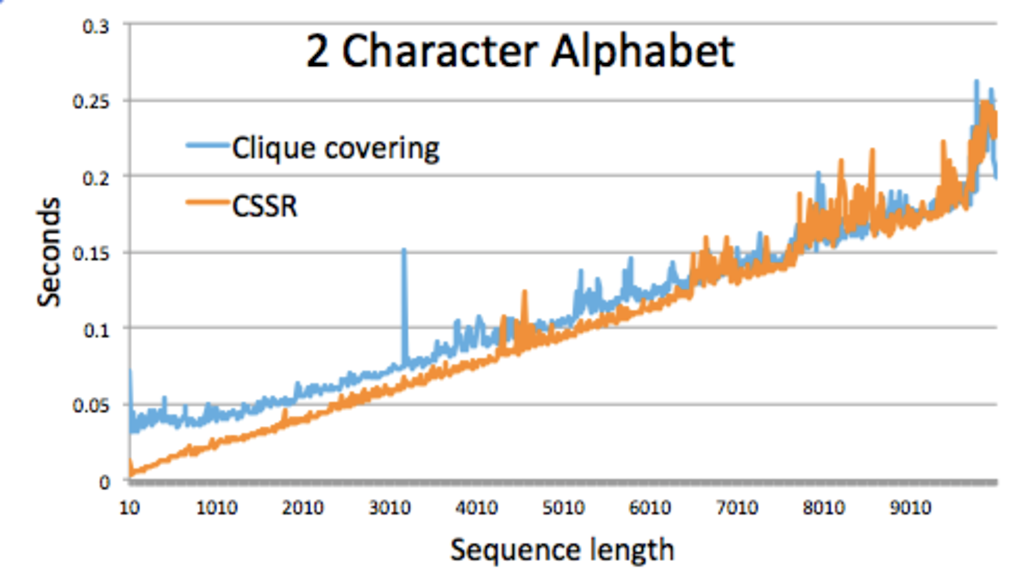}
\end{subfigure}
\begin{subfigure}
	\centering
	\includegraphics[scale=.38]{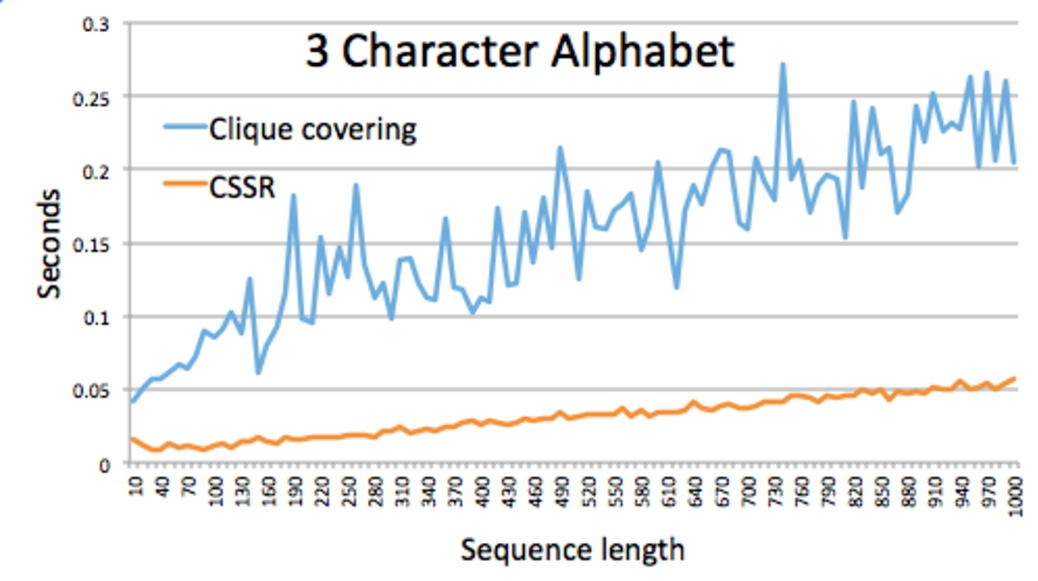}
\end{subfigure}
\begin{subfigure}
	\centering
	\includegraphics[scale=.38]{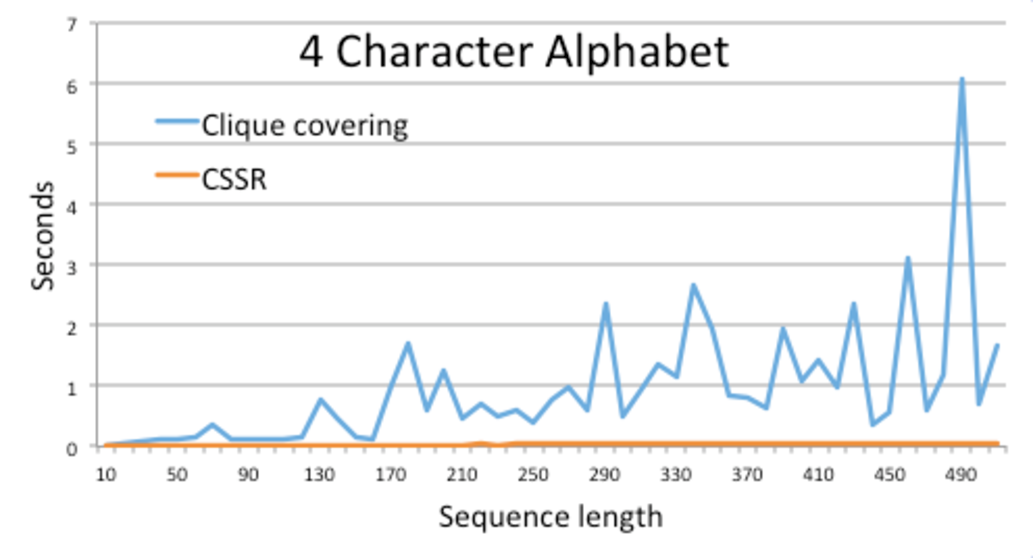}
\end{subfigure}
\caption{Minimal clique covering run times compared to CSSR}
\label{fig:charactercompare}
\end{figure}

\section{Conclusion and Future Directions}
In this paper we illustrated the problem stated by Shalizi and Crutchfield \cite{Shal02} of determining a minimal state probabilistic finite state representation of a data set is $\mathrm{NP}$-hard for finite data sets. As a by-product, we formally proved this to be the case for the non-unifilar case as studied in \cite{Schm06}. As such, this shows that both the CSSR algorithm of \cite{Shal02} and CSSA algorithm of \cite{Schm06} can be thought of as low-order polynomial approximations of $\mathrm{NP}$-hard problems.

Future work in this area includes studying, in detail, these approximation algorithms for the problems to determine what their approximation properties are. \ecp{A specific area of interest is placing bounds on the number of states given by a solution to Equation \ref{eqn:Problem} based on properties of the input string. We plan to investigate other linear approximation algorithms like those found in \cite{Carrasco,Higuera}.}

In addition to this work, determining weaknesses in this approach to modeling behavior are planned. We are particularly interested in the affects deception can have on models learned in this way. For example, we are interested in formulating a problem of constrained optimal deception in which a learner, using an algorithm like the one described here or the CSSR or CSSA approximations, is optimally confused by an input stream that is subject to certain constraints in its incorrectness.


\section*{Acknowledgements}
Portions of Dr. Griffin's and Ms. Paulson's work were supported by the Army Research Office under Grant W911NF-11-1-0487.

\appendix
\section{Computing $f_{q_i|\mathbf{y}}$ and
$f_{\mathbf{x}|\mathbf{y}}$}

The following formulas can be used to compute $f_{q_i|\mathbf{y}}$ and
$f_{a\mathbf{x}|\mathbf{y}}$ in Algorithm \ref{alg:CSSR}. Let
$\#(\mathbf{x},\mathbf{y})$ be the number of times the sequence
$\mathbf{x}$ is observed as a subsequence of $\mathbf{y}$.

\begin{equation}
f_{\mathbf{x}|\mathbf{y}}(a) =
\Pr(a | \mathbf{x}, \mathbf{y}) = \frac{\#(\mathbf{x}a,\mathbf{y})}
                                        {\#(\mathbf{x},\mathbf{y})}
\label{eqn:StringProb}
\end{equation}

\begin{equation}
f_{q_i|\mathbf{y}}(a) =
\Pr(a | q_i, \mathbf{y}) =
    \frac{\sum_{\mathbf{x} \in q_i}
            \#(\mathbf{x}a,\mathbf{y})}
            {\sum_{\mathbf{x} \in q_i}\#(\mathbf{x},\mathbf{y})}
\label{eqn:StateProb}
\end{equation}


\bibliographystyle{IEEEtran}
\bibliography{IEEEfull,Refs}

\end{document}